\documentclass[onecolumn]{IEEEtran}
\usepackage{calc}
\usepackage[T1]{fontenc}

\usepackage[overload]{empheq}
\usepackage{cleveref}
\usepackage{multirow}
\usepackage{cite}
\usepackage{graphicx,subfigure}
\usepackage{psfrag}
\usepackage{amsmath,amssymb}
\usepackage{color}
\usepackage{tikz}
\usepackage{cleveref}
\usepackage{empheq}
\usepackage{setspace}
\newenvironment{proof}{\begin{IEEEproof}}{\end{IEEEproof}}

\DeclareMathOperator*{\defeq}{\triangleq}

\newtheorem{proposition}{Proposition}
 
\newcommand{\bit}{\begin{itemize}}
\newcommand{\eit}{\end{itemize}}

\newcommand{\bc}{\begin{center}}
\newcommand{\ec}{\end{center}}

\newcommand{\ba}{\begin{array}}
\newcommand{\ea}{\end{array}}

\newcommand{\beq}{\begin{equation}}
\newcommand{\eeq}{\end{equation}}

\newcommand{\beqn}{\begin{equation*}}
\newcommand{\eeqn}{\end{equation*}}

\newcommand{\bean}{\begin{eqnarray*}}
\newcommand{\eean}{\end{eqnarray*}}
\newcommand{\bea}{\begin{eqnarray}}
\newcommand{\eea}{\end{eqnarray}}

\def\F{\mathbb{F}}

\def\av{\boldsymbol{a}}

\def\cv{\boldsymbol{c}}

\def\xv{\boldsymbol{x}}
\def\yv{\boldsymbol{y}}

\newtheorem{remark}{Remark}

\makeatletter
\def\blfootnote{\gdef\@thefnmark{}\@footnotetext}
\makeatother

\begin{document}
\sloppy

\title{Improved Latency-Communication Trade-Off for Map-Shuffle-Reduce Systems with Stragglers}
\author{ Jingjing Zhang and Osvaldo Simeone  
	}  
\maketitle

\thispagestyle{empty}

\begin{abstract}
In a distributed computing system operating according to the map-shuffle-reduce framework, coding data prior to storage can be useful both to reduce the latency \blfootnote{The authors are with the Department of Informatics, King's College London, London, UK (emails: jingjing.1.zhang@kcl.ac.uk, osvaldo.simeone@kcl.ac.uk).}caused by straggling servers and to decrease the inter-server communication load in the shuffling phase. In prior work, a concatenated coding scheme was proposed for a matrix multiplication task. In this scheme, the outer Maximum Distance Separable (MDS) code is leveraged to correct erasures caused by stragglers, while the inner repetition code is used to improve the communication efficiency in the shuffling phase by means of coded multicasting. In this work, it is demonstrated that it is possible to leverage the redundancy created by repetition coding in order to increase the rate of the outer MDS code and hence to increase the multicasting opportunities in the shuffling phase. As a result, the proposed approach is shown to improve over the best known latency-communication overhead trade-off.
\end{abstract}


\section{introduction}

Consider the distributed computing system shown in Fig.~\ref{fig:model}, in which $K$ servers are tasked with the computation of the matrix product $\mathbf{Y}=\mathbf{A X}$, where data matrix $\mathbf{X}$ is available at all servers, while the matrix $\mathbf{A}$ can be partially stored at each server. In particular, each server can store information about $\mathbf{A}$  up to a fraction $\mu \leq 1$ of its size. As a result, servers need to collaborate in order to compute the output $\mathbf{Y}$ by communicating over a shared multicast channel. Furthermore, servers are typically subject to random computing times, and hence measures should be taken in order to ensure correct distributed computation even in the presence of a given number of straggling servers \cite{DB:13}.

A standard framework to implement this operation is \emph{map-shuffle-reduce} \cite{LAA:15,LAA:18}. Accordingly, in the \emph{map phase}, the servers compute Intermediate Vaues (IVs) that depend on the available information about $\mathbf{A}$ and $\mathbf{X}$. Generally, only a subset of $q$ servers is able to complete this phase within a desired latency. In the \emph{shuffling phase}, functions of the IVs are exchanged among the $q$ non-straggling servers. Finally, in the \emph{reduce phase}, the non-straggling servers can collectively produce all the columns in $\mathbf{Y}$, with each server producing a given subset (see Fig.~\ref{fig:model}).

The performance of the system is characterized by a trade-off between computational latency --- the time elapsed during the map phase --- and communication overhead --- the amount of information exchanged during the shuffling phase \cite{MMRA:16}. In fact, waiting for a longer time for more servers to compute their map operations reduces the need for communication of IVs during the shuffling phase. This trade-off depends on the available storage capacity $\mu$ at the servers, which limits the capability of the servers to compute IVs and to withstand erasures due to stragglers \cite{LMA:16}.   

In prior work \cite{LMA:16}, a concatenated coding scheme was proposed for the described matrix multiplication task (see Fig.~\ref{fig:redundancy}). In this scheme, the outer Maximum Distance Separable (MDS) code is leveraged to correct erasures caused by stragglers, while the inner repetition code is improve the communication efficiency in the shuffling phase by means of coded multicasting. In this paper, it is demonstrated that it is possible to leverage the redundancy created by repetition coding in order to increase the rate of the outer MDS code and hence to enhance the multicasting opportunities in the shuffling phase. As a result, the proposed approach is shown to improve over the best known latency-communication overhead trade-off (see Fig.~\ref{fig:com} for a preview).

The rest of the paper is organized as follows. Section~\ref{sec:2} describes a distributed computing model with straggling servers for map-shuffle-reduce systems and reviews the unified coding scheme studied in \cite{LMA:16}. In Section~\ref{sec:3}, we propose an improved concatenated coding scheme. In particular, an illustrative example is presented first, and then the achievable communication load is characterized, followed by the corresponding general scheme. Section~\ref{sec:4} concludes this work. 

\textbf{Notation:}
For $a\in \mathbb{N}^+, b\in\mathbb{Z}$, we define $\binom{a}{b}=0$ when $a<b$ or $b<0$, and we define $\binom{a}{0}=1$. For $K,P\in \mathbb{N}^+$ with $K\leq P$, we define the set $[P]\defeq\{1,2,\cdots,P\}$, and the set $[K : P]\defeq \{K, K + 1, \cdots, P\}$. For a set $\mathcal{A}$, $|\mathcal{A}|$ represents the cardinality. We also have $0/0=0$. Matrices and vectors will be denoted by upper-case and lower-case bold font, respectively. 

\section{system Model and background} \label{sec:2}

\subsection{System Model}
Consider a distributed implementation of the matrix multiplication task described by the equality
\begin{align} \label{product}
\mathbf{Y}=\mathbf{A X},
\end{align}
with the task-specific matrix $\mathbf{A}\in \F_{2^T}^{m \times n}$ and the input data matrix $\mathbf{X}\in \F_{2^T}^{n \times N}$, where each element of matrices $\mathbf{A}$ and $\mathbf{X}$ consists of $T$ bits, and we have the parameters $T, m,n, N\in \mathbb{N}^+$. We use $\xv_i\in\F_{2^T}^{n}$ and $\yv_i\in\F_{2^T}^{m}, i\in[N]$ to denote each column vector of input matrix $\mathbf{X}$ and output matrix $\mathbf{Y}$, respectively. Hence, the matrix product \eqref{product} corresponds to the $N$ linear equations $\yv_i=\mathbf{A}\xv_i$, for $i\in[N]$.

There are $K$ distributed servers, each having a storage device of size $\mu mnT$ bits, with $\mu \in [1/K,1]$. Hence, each server $k$, with $k\in[K]$, can store a number of bits equal to a fraction $\mu$ of the size of matrix $\mathbf{A}$. The lower bound $1/K$ ensures that the entire matrix can be stored across all servers. Specifically, we assume that each server stores up to $m\mu$ row vectors selected from the rows $\mathcal{C}=\{\cv_i\}_{i=1}^{m'}$ of the linearly encoded matrix 
\begin{align} \label{encode}
\mathbf{C}=[\cv_1^T, \cdots,\cv_{m'}^T]^T=\mathbf{G}\mathbf{A},
\end{align} 
where we have defined the encoding matrix $\mathbf{G}\in \F_{2^T}^{m' \times m}$, with integer $m'\geq m$. The rows stored by server $k$ are described by the set $\mathcal{C}_k\subseteq\mathcal{C}$, with $|\mathcal{C}_k|\leq m\mu$. Furthermore, as seen in Fig.~\ref{fig:model}, each server is assumed to have the entire data matrix $\mathbf{X}$ available. Servers have random computations times, and hence they may be straggling. Furthermore, they can communicate to each other over multicast links. 

Following the standard map-reduce paradigm, the computation process consists of three phases, namely map, shuffling and reduce. The goal is to have all columns $\{\yv_i\}_{i=1}^{N}$ available at the end of the reduce phase at the subset of servers that have completed computations within some tolerable computation time. As an example, in Fig.~\ref{fig:model}, only server 1 and 2 are available, while server 3 is straggling, and server $i$ produces subset $\mathcal{R}_i$ of columns $\{\yv_i\}_{i=1}^{N}$, with $\mathcal{R}_1\cup\mathcal{R}_2=\{\yv_i\}_{i=1}^{N}$.

\begin{figure}[t!] 
  \centering
\includegraphics[width=0.5\columnwidth]{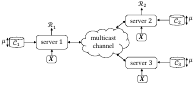}
\caption{Distributed computing platform with $K=3$ servers tasked with producing the matrix product \eqref{product} with minimal latency and inter-server communication load. In this example, server 3 is straggling in the map phase and hence the output is produced by servers 1 and 2.}
\label{fig:model}
\end{figure}

\emph{Map phase:} Each server $k$, with $k\in[K]$, computes the $m\mu$ products 
\begin{align} \label{im}
\mathcal{I}_k=\{\cv\mathbf{X} \in \F_{2^T}^{1 \times N}:\cv \in \mathcal{C}_k\},
\end{align}
for all stored encoded rows $\cv\in\mathcal{C}_k$. The contents of set $\mathcal{I}_k$ are referred to as the IVs available at server $k$ using the map-reduce terminology. We define $D(q)$ as the average time required for the first $q$ servers to complete their computations. The set $\mathcal{Q}\in[K]$, with $|\mathcal{Q}|=q$, of $q$ servers is arbitrary and function $D(q)$ is determined by the distribution of the random computation times of the servers. As a specific example, if each server requires a time distributed as a shifted exponential with minimum value $\mu N$ and average $2 \mu N$, i.e., with cumulative distribution function $F(t)=1-e^{-(t/(\mu N)-1)}$, for all $t \geq \mu N$,
then the computation latency $D(q)$ can be derived as \cite{LMA:16}
\begin{align} \label{dq}
D(q)=\mu N\bigg(1+\sum_{j=K-q+1}^{K}\frac{1}{j}\bigg). 
\end{align}
Note that the average latency $D(q)$ increases with $q$ and with the product $\mu N$. The latter is proportional to the computation load at the server during the map phase, due to the computation of $\mu mN$ inner products in \eqref{im}.

\emph{Shuffling phase:} After the average time $D(q)$, all servers in the non-straggling set $\mathcal{Q}$ coordinate by assigning each server $k$ in $\mathcal{Q}$ a subset of the vectors $\{\yv_i\}_{i=1}^{N}$ to be computed. The indices of the vectors "reduced" at server $k$ are described by the set $\mathcal{R}_k\subseteq[N]$, with $\bigcup_{k\in\mathcal{Q}}\mathcal{R}_k=[N]$. To enable each server $k$ to reconstruct the vectors $\{\yv_i:i\in \mathcal{R}_k\}$, any functions of the computed IVs \eqref{im} can be exchanged among the $q$ servers in subset $\mathcal{Q}$ during the shuffling phase. We use $\mathcal{M}_k$ to denote the sets of functions of the IVs $\mathcal{I}_k$ that each server $k$ multicasts to all other servers in $\mathcal{Q}$ during the shuffling phase.

\emph{Reduce phase:} With the received data $\big\{\mathcal{M}_{k'}:k'\in\mathcal{Q}\backslash \{k\}\big\}$ multicast by the other servers and with the locally computed IVs $\mathcal{I}_k$, each server $k$ in $\mathcal{Q}$ computes the assigned vectors $\{\yv_i:i\in \mathcal{R}_k\}$ in the reduce phase. 

\emph{Performance Criteria:} For a value of the number $q$, of non-straggling servers, the total number of bits exchanged among the servers during the shuffling phase is $\sum_{k\in\mathcal{Q}}|\mathcal{M}_k|$. With normalization by the number of output bits $mT$ for a column, we define the communication load as
\begin{align}
L(q)=\frac{\sum_{k\in\mathcal{Q}}|\mathcal{M}_k|}{mT}.
\end{align} 
Given a computation latency function $D(q)$, e.g. \eqref{dq}, a pair $\big(D(q),L(q)\big)$ is said to be achievable if there exists feasible map, shuffling and reduce policies, i.e., a specific construction for matrix $\mathbf{G}$, communication strategy $\{\mathcal{M}_k:k\in\mathcal{Q}\}$, and reduce task assignment $\{\mathcal{R}_k:k\in\mathcal{Q}\}$, that ensures a correct reconstruction of all columns of \eqref{product} across the $q$ non-straggling servers. Finally, we define the optimal latency-load trade-off curve as
\begin{align}
L^*(q)=\inf\{L(q): \big(D(q),L(q)\big) ~\text{is achievable for sufficiently large}~ m, N, ~\text{and}~T\}. \label{def:to}
\end{align} 

\subsection{Background}

In \cite{LMA:16}, a map-shuffling-reduce coding scheme is introduced that is based on the concatenation of two codes: an MDS code of rate $r_1$ and a repetition code of rate $r_2$. The $(Km/q,m)$ MDS code of rate 
\begin{align} \label{fixedr1}
r_1=\frac{K}{q}
\end{align} 
is used in order to enable decoding from the output of an arbitrary set of $q$ non-straggling servers. To this end, the $K-q$ unavailable outputs from the straggling servers, which have not completed their map computations within average time $D(q)$, are treated as erasures. The repetition code of rate 
\begin{align} \label{fixedr2}
r_2=q\mu
\end{align}
is instead used for the purpose of reducing the inter-server communication load during the shuffling phase. This is done by leveraging coded multicasting based on the available side information at the servers \cite{LAA:15}. The overall process for the coding strategy used in the map phase is illustrated in Fig.~\ref{fig:redundancy}.

\begin{figure}[t!] 
  \centering
\includegraphics[width=0.45\columnwidth]{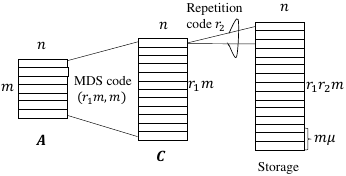}
\caption{Encoding scheme used in the map phase.}
\label{fig:redundancy}
\end{figure}

Since each server can store at most $m\mu$ distinct encoded rows by the storage constraint, the overall storage redundancy across all the $K$ servers with respect to the $m$ rows of matrix $\mathbf{A}$ is given as $Km\mu/m=K\mu$. The map code in Fig.~\ref{fig:redundancy} splits this redundancy between the above two codes as 
\begin{align} \label{con:storage}
K\mu=r_1r_2,
\end{align} 
where $r_1$ and $r_2$ are selected in \cite{LMA:16} as \eqref{fixedr1} and \eqref{fixedr2}, respectively. 

The scheme in \cite{LMA:16} uses the MDS code, of rate \eqref{fixedr1}, solely as a latency-reducing code in order to correct erasures due to stragglers, while the repetition code, of rate \eqref{fixedr2}, is exclusively leveraged to accelerate communications as a bandwidth reducing code for coded multicasting. To elaborate, in \cite{LMA:16}, from the point of view of the MDS code, the effective storage of each server is $\mu_{\mathrm{eff}}=\mu/r_2=1/q$ due to the repetition coding rate $r_2$. This is in the sense that the MDS code effectively "sees" a storage capacity equal to $\mu_{\mathrm{eff}}$, given that the rest of the capacity is used to store repetitions of MDS-encoded rows. As a result, each server effectively stores $m\mu_{\mathrm{eff}}=m/q$ rows, and, with the $\big(K\times(m/q),q\times(m/q)\big)$ MDS code, decoding can be successful as long as $q$ servers complete computing. In fact, the MDS code can correct $K-q$ erasures, each corresponding to the $m/q$ unavailable coded rows stored at each straggling server. 

%



\section{Improved Coding Scheme Based on Concatenated Coding}  \label{sec:3}

%
%

In this section, we propose a policy that is based on the idea of leveraging the overall $(rm,m)$ concatenated code with redundancy $r=r_1r_2$ in Fig.~\ref{fig:redundancy} for the purpose of correcting erasures due to stragglers. In essence, the repetition code of rate $r_2$ is used not only as a bandwidth-reducing code, but it also contributes to the reduction of latency in the map phase. As we will see, this allows us to establish a set of feasible choices of $r_1$ and $r_2$ that contain the selection on \eqref{fixedr1}-\eqref{fixedr2} as special case. 

As a result, the proposed approach allows the section of a rate $r_1\leq K/q$ of the MDS code, which is generally lower than \eqref{fixedr1}, and a storage redundancy $r_2\geq q\mu$ of the repetition code, which is generally larger than \eqref{fixedr2}, while still satisfying the storage constrain $r_1r_2\leq K\mu$ (cf.~\eqref{con:storage}), and the achievability requirement. Therefore, as compared to the algorithm in \cite{LMA:16}, the proposed scheme can potentially operate with a lower rate $r_1$ and a higher rate $r_2$ and is hence potentially able to further speed up the shuffling phase, reducing the communication overhead $L(q)$. 

To proceed, we first present an illustrative example; then, we provide a characterization of the achievable communication load $L(q)$; and, finally, we describe the corresponding general scheme. 


\subsection{Illustrative Example} \label{example}
We now present an example to illustrate the proposed concatenated coding policy. Consider $K=6$ servers, with $q=4$ non-stragglers, storage capacity $\mu=1/2$, as well as the parameters $m=20$ and $N=12$. Parameter $T$ is sufficiently large so as to ensure the existence of an MDS code \cite{LC}. For this example, the scheme in \cite{LMA:16} chooses $r_1=3/2$ and $r_2=2$ according to \eqref{fixedr1}-\eqref{fixedr2}. As shown below, the proposed scheme instead can operate with $r_1=1$ and $r_2=K\mu=3$. Accordingly, the MDS code is not used and the scheme only relies on the repetition code of rate $r_2$ to correct erasures. We will see that the larger value of $r_2$ allows for a reduction of the communication load $L(q=4)$ as compared to the selection in \cite{LMA:16}. 

\emph{Map phase:} Without using an MDS code to encode the rows of matrix $\mathbf{A}$, i.e., with $r_1=1$, we have $\mathbf{C}=\mathbf{A}$ in \eqref{encode}. Each row of matrix $\mathbf{A}$ is then replicated $r_2=3$ times, so that each server $k$ stores $|\mathcal{C}_k|=\mu m=10$ uncoded rows of $\mathbf{A}$. This is done by storing each row in a subset $\mathcal{K}\subseteq[K]$ of three servers, with $|\mathcal{K}|=r_2$. We write $\av_{\mathcal{K}}$ for the row that is stored at all servers in set $\mathcal{K}$ and we have $\mathcal{C}_k=\{\av_{\mathcal{K}}: k\in\mathcal{K}\}$ for the set of rows stored at server $k$. Without loss of generality, we assume that servers 1, 2, 3, 4 are the first $q=4$ servers that complete their computations, i.e., $\mathcal{Q}=\{1,2,3,4\}$. We recall that each server $k$ computes the $|\mathcal{C}_k|=10$ products in the set $\mathcal{I}_k$ in \eqref{im} in the map phase. For the reduce phase, since there are $N=12$ output vectors $\{\yv_i\}_{i=1}^{N}$ in matrix $\mathbf{Y}$, each server $k$ is assigned to output three consecutive vectors, i.e., $\mathcal{R}_k=\{\yv_{3(k-1)+i}=\mathbf{A}\xv_{3(k-1)+i}: i\in[3]\}$.



\emph{Shuffling phase.}
To this end, in the shuffling phase, each server needs to obtain a set of IVs through multicast transmissions. Take server 1 for example. To reduce vector $\yv_1=\mathbf{A}\xv_1 \in \mathcal{R}_1$, server 1 can use the IVs $\{\av\xv_1: \av\in \mathcal{C}_1 \}$ in $\mathcal{I}_1$. Hence, it needs ten extra IVs $\{\av\xv_1: \av \in \mathcal{C}\backslash \mathcal{C}_1 \}$, each of which has been computed by at least one of the servers 2, 3, 4. This is because each row $\av \notin \mathcal{C}_1$ is stored at $r_2=3$ servers that do not include server 1. The same holds for vectors $\yv_2$ and $\yv_3$, and thus server 1 needs 30 IVs in total from servers 2, 3, 4. Similarly, each of the other three servers requires 30 IVs in shuffling phase. We emphasize that, thanks to the repetition code, all the necessary information for the reduce phase is available at any $q$ servers. 

In order to ensure that each server receives the described IVs, we operate separately by multicasting messages within groups of $i+1$ servers in three different phases, which are indexed as $i=3,2,1$, and carried out in this order. This follows the same approach as in \cite{LMA:16}, with the caveat that here we can benefit from a larger multicasting gain (see Remark~\ref{remark1}).

Accordingly, in the first phase, labeled as $i=3$ and illustrated in Fig.~\ref{fig:shuffle}(a), the servers share the needed IVs that are available at subsets of three of the four servers in $\mathcal{Q}$. By construction, one such message exists at each of the four disjoint subsets of three servers in $\mathcal{Q}$. To this end, we perform coded shuffling among the four servers by sending a multicasting message from one server to the other three. As a result of each transmission, each of the receiving server can recover one desired IVs. As illustrated in Fig.~\ref{fig:shuffle}(a), for example, server 1 sends message $\av_{123}\xv_{10}\oplus\av_{124}\xv_7\oplus\av_{134}\xv_4$. Then, server 2 can recover $\av_{134}\xv_4$ by canceling the IVs $\av_{123}\xv_{10}$ and $\av_{124}\xv_7$ that are available in $\mathcal{I}_2$. Similarly, server 3 can recover $\av_{124}\xv_7$ and server 4 can obtain $\av_{123}\xv_{10}$. A similar procedure applies for the other three multicasting messages. At the end of this phase, each server can obtain three IVs.

\begin{figure*}[t!]
    \centering
    \begin{subfigure}
        \centering
        \includegraphics[width=0.4\columnwidth]{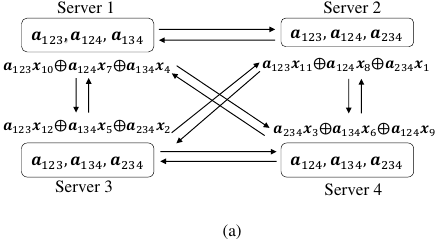}
    \end{subfigure}%
    ~ 
    \begin{subfigure}
        \centering  
        \includegraphics[width=0.45\columnwidth]{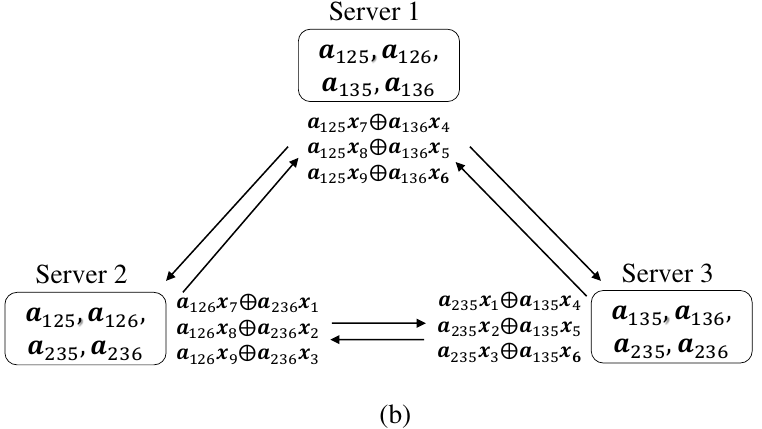} 
    \end{subfigure}
    \caption{Illustration of two phases of shuffling for the example in Section~\ref{example}: (a) Coded shuffling among servers 1, 2, 3 and 4 in the first phase $(i=3)$; (b) Coded shuffling among servers 1, 2, 3 in the second phase $(i=2)$.}
		\label{fig:shuffle}
\end{figure*}


In the second phase, for $i=2$, we deliver the needed IVs that are available at subsets of two of the four server in $\mathcal{Q}$. By construction, there are 72 such IVs, with six available at each of the four subsets of three servers in $\mathcal{Q}$. The four groups of three servers operate in the same way by sending a set of multicasting messages from one server to the other two. Take for example the subset $\{1,2,3\}$ illustrated in Fig.~\ref{fig:shuffle}(b). Server 2 sends the three messages $\{\av_{126}\xv_{6+i}\oplus\av_{236}\xv_{i}, i\in[3]\}$, from which server 1 and 3 obtain $\{\av_{236}\xv_{i}, i\in[3]\}$ and $\{\av_{126}\xv_{6+i}, i\in[3]\}$, respectively, by canceling their own interference with side information.  Server 1 and 3 can perform coded multicasting in a similar manner. As a result, for any subset of three servers, each server can obtain $6$ needed IVs. After this phase, each server can recover 18 IVs in total.

Finally, each server $k$ still needs 9 IVs $\{\av_{i56}\xv_{3(k-1)+j}:i\in[4]\backslash\{k\}, j\in[3]\}$, each of which is computed by only one of the four servers in $\mathcal{Q}$. Hence, in the last phase, labeled as $i=1$, the overall $36$ IVs are shared by means of unicast transmission.

%
%

To sum up, $4+36+36=76$ coded IVs are communicated sequentially in the shuffling phase. This yields a communication load of $L(q=4)=76/m=3.8$, which is smaller than that of $L(q=4)=4.2$ in \cite{LMA:16}.

\begin{remark} \label{remark1}
This example shows that the MDS code can be avoided when $q$ is sufficient large. This is because, thanks to the repetition code, the entire matrix $\mathbf{A}$, and hence the product \eqref{product}, can be recovered by combining the information available at the non-straggling servers when we store uncoded rows from matrix $\mathbf{A}$. As compared to the policy in \cite{LMA:16}, a higher storage redundancy $r_2=K\mu=3$ is obtained for the repetition code. As a result, while the maximum multicasting gain in the shuffling phase for \cite{LMA:16} is limited to 2, i.e., multicast messages can be sent to groups of servers of size at most 2, in the proposed algorithm, the multicasting gain $K\mu=3$ can be reaped, achieving a lower communication load. 
\end{remark}

\subsection{Main result}

In this section, we generalize the scheme introduced in the previous example. As described in Fig.~\ref{fig:redundancy}, in the map phase, we use an $(r_1r_2m, m)$ concatenated code to encode matrix $\mathbf{A}$. The first is an MDS code with rate $r_1=l/q$, for some integer $l\in[q:K]$, and the second is a repetition code, with integer rate $r_2 \leq K\mu/r_1$. Note that the storage constraint $r_1r_2\leq K\mu$ (cf.~\eqref{con:storage}) is satisfied, and hence each server stores at most $m\mu$ distinct encoded rows. In the next proposition, we identify sufficient conditions on the rate pairs $(r_1,r_2)$ to yield a feasible policy.
 
\begin{proposition} \label{Pro:condition}
For storage capacity $\mu\in[1/K,1]$ and number of non-straggling servers $q$ with $ q \in[\lceil 1/\mu \rceil: K]$, sufficient conditions for rates $(r_1,r_2)$ to yield a feasible policy are 
\begin{subequations} \label{condition}
\begin{align}
&~q r_1\in[q:K],~ r_2\in[\lfloor q\mu \rfloor: \lfloor K\mu \rfloor], \label{cons:lim} \\
&~r_1 r_2\leq K\mu, ~\text{and} \label{cons:stor} \\
&\binom{K}{r_2}-\binom{K-q}{r_2}\geq\frac{1}{r_1}\binom{K}{r_2}. \label{cons:numr} 
\end{align}
\end{subequations}
\end{proposition} 
\begin{proof}
The proof is presented in Section~\ref{proof}. 
\end{proof}

\begin{remark}
Condition \eqref{cons:lim} defines the domains of rates $r_1$ and $r_2$. Condition \eqref{cons:stor} impose the storage capacity constraint (cf.~\eqref{con:storage}), while condition \eqref{cons:numr} ensures the feasibility of the reconstruction requirement of data matrix $\mathbf{A}$. It can be verified that the choice \eqref{fixedr1} and \eqref{fixedr2} in \cite{LMA:16} satisfies all conditions \eqref{condition} (see Appendix). Furthermore, when $K-q< \lfloor K\mu \rfloor$, a feasible choice is $(r_1=1,r_2=\lfloor K\mu \rfloor)$. With this choice, the MDS code can be avoided, as shown in the example of Section~\ref{example}. 
\end{remark}

Using a feasible pair of rates $(r_1,r_2)$ satisfying \eqref{condition}, the followed communication load is achievable.

\begin{proposition} \label{ach}
For a matrix multiplication task executed by $K$ distributed servers, each having a fractional storage size $\mu\in [1/K,1]$, the following communication load is achievable in the presence of $K-q$ straggling servers with $q \in[\lceil 1/\mu \rceil:K]$
\begin{align} 
\underset{(r_1, r_2)}{\text{min}}~~&\Bigg(L(q)=N\sum_{j=s_q}^{s_{max}} \frac{B_j}{j} +\frac{N(1-\frac{r_1r_2}{K}-\sum_{j=s_q}^{s_{max}}B_j)}{s_q-1} \Bigg)\label{load} 
\end{align} 
where pair $(r_1,r_2)$ is constrained to satisfy the feasible condition \eqref{condition} and we have defined
\begin{align} 
s_{max} \defeq \min\{q-1,r_2\}, B_j\defeq \frac{\binom{q-1}{j}\binom{K-q}{r_2-j}}{\frac{1}{r_1}\binom{K}{r_2}},s_q\defeq \inf\bigg\{s:\sum_{j=s}^{s_{max}}B_j\leq 1-\frac{r_1r_2}{K}\bigg\}.
\end{align}
\end{proposition}

\begin{proof}
The proof is presented in Section~\ref{proof}. 
\end{proof}
\begin{remark}
Since the solution \eqref{fixedr1} and \eqref{fixedr2} of \cite{LMA:16} is always feasible for the constraints \eqref{condition}, it follows that the achievable communication load $L(q)$ in \eqref{load} is always no larger than that in \cite{LMA:16}. This is because the load \eqref{load} with $r_1$ and $r_2$ in \eqref{fixedr1}-\eqref{fixedr2} is no greater than the load \cite[eq.~(9)]{LMA:16}. To demonstrate that the improvement can be strict, beside the example given in Section~\ref{example}, we provide here a numerical example. 
\end{remark}


\begin{figure}[t!] 
  \centering
\includegraphics[width=0.5\columnwidth]{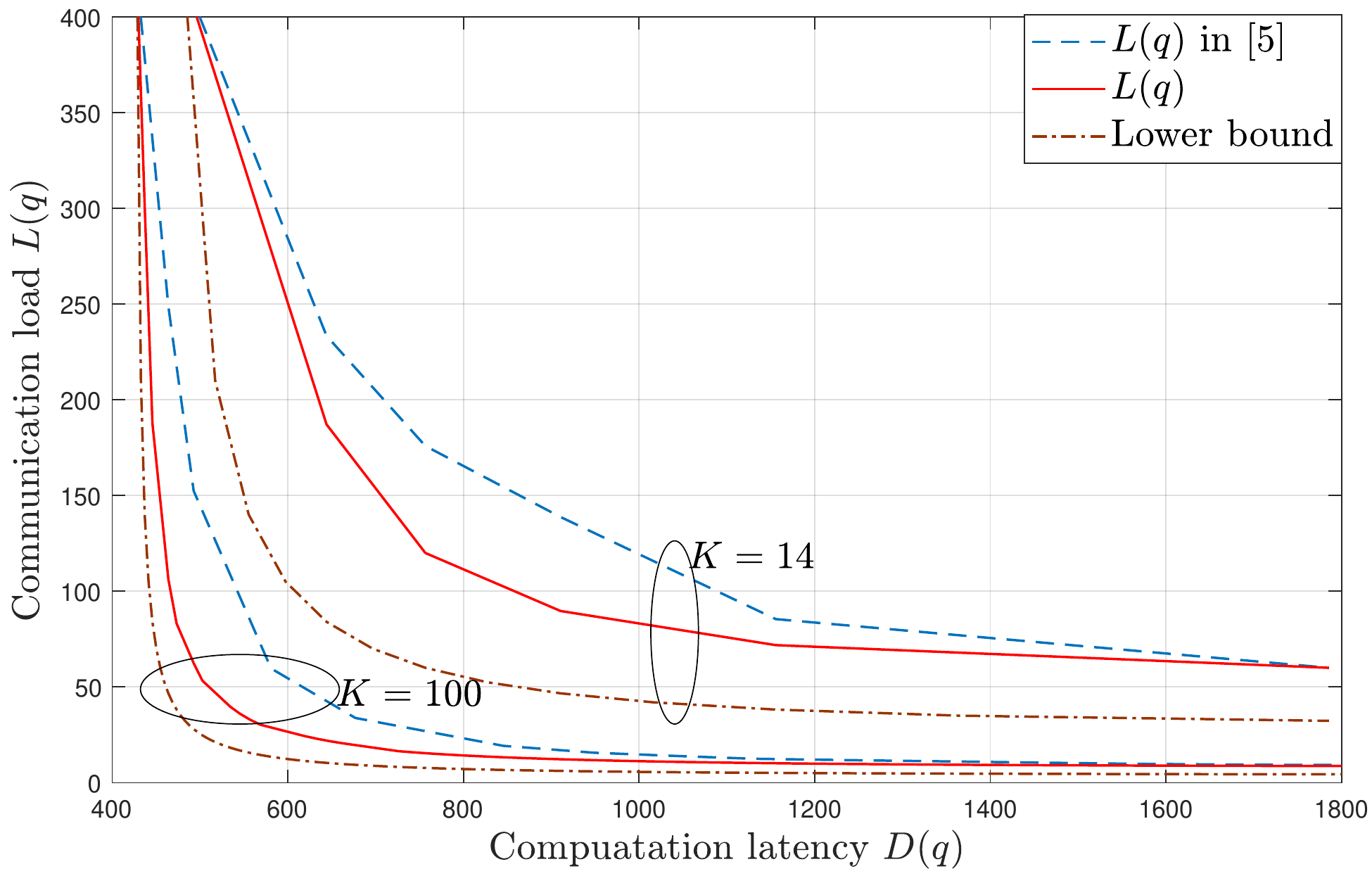}
\caption{Achievable loads $L(q)$ as a function of the computation latency $D(q)$ for the proposed scheme and for the scheme in \cite{LMA:16}, along with lower bound in \cite{LMA:16}, with $N=840,\mu=1/2$ and different values of $K$.}
\label{fig:com}
\end{figure}

A comparison of the achievable communication loads of the proposed scheme and of the scheme in \cite{LMA:16}, as well as the lower bound in \cite{LMA:16}, can be found in Fig.~\ref{fig:com}, where we apply the computation latency $D(q)$ modeled as in \eqref{dq} and we have set $N=840,\mu=1/2$, and different values of $K$. For each value of $K$, at the two end points of small or large computation latency $D(q)$, obtained with $q=\lceil 1/\mu \rceil$ and $q=K$, respectively, the two achievable communication loads coincide. The proposed scheme is seen to bring a positive reduction in communication load, as compared to the algorithm in \cite{LMA:16} for intermediate values of the latency $D(q)$, that is, when there are a moderate number of straggling servers. Furthermore, as the total number $K$ of servers increases, it is observed that the reduction is more significant, and the gap between the achievable load of the proposed scheme and the lower bound becomes smaller. As a numerical example, for $K=100$ and $D(q)=600$, the proposed scheme reduces the communication load by a factor of 2 around, and for $D(q)=500$ by a factor of 2.5.

\subsection{General Scheme and Proofs of Propositions \ref{Pro:condition} and \ref{ach}} \label{proof}

We now generalize the map-shuffling-reduce policy introduced in Section~\ref{example} and we prove that it is feasible for any rate pair $(r_1,r_2)$ that satisfies all conditions $\eqref{condition}$, while achieving a communication load given by $L(q)$ in \eqref{load}. The strategy follows the approach in \cite{LMA:16}, but it relies on a more general choice of values for the rate pair $(r_1,r_2)$.

\emph{Map phase:} Choose values of rates $r_1$ and $r_2$ within their respective domain according to condition \eqref{cons:lim}. As seen in Fig.~\ref{fig:redundancy}, we use an $(r_1m,m)$ MDS code, with a sufficiently large $T$, in order to encode the $m$ rows of matrix $\mathbf{A}$ into $r_1m$ coded rows $\{\cv_i\}_{i=1}^{r_1m}$. Then, we partition all the $r_1m$ encoded rows into $\binom{K}{r_2}$ disjoint subsets, each indexed by a subset $\mathcal{K}\subseteq[K]$ of size $r_2$, i.e., $|\mathcal{K}|=r_2$. Each subset of encoded rows, denoted as $\mathcal{C}_{\mathcal{K}}$, contains $(r_1m)/\binom{K}{r_2}$ encoded rows. Furthermore, for each server $k$, we define the stored subset of encoded rows as $\mathcal{C}_k=\bigcup_{\mathcal{K}: k\in\mathcal{K}} \mathcal{C}_{\mathcal{K}}$ so that a row $\cv$ is included in server $k$ if it belongs to $\mathcal{C}_{\mathcal{K}}$ with $k\in\mathcal{K}$. As a result, each server stores $|\mathcal{C}_k|=|\mathcal{C}_{\mathcal{K}}|\binom{K-1}{r_2-1}=r_1 r_2 m/K$ rows. Note that $m$ should be sufficient large so that $|\mathcal{C}_k|$ is an integer. We also need to impose the inequality $r_1 r_2 m/K\leq m\mu$, in order to satisfy the capacity constraint of each server. This corresponds to inequality \eqref{cons:stor}.

In the map phase, each server $k$, with $k\in[K]$, computes the IVs in set $\mathcal{I}_k$ defined in \eqref{im}. In order to ensure feasibility, i.e., to guarantee the correctness of computation \eqref{product}, we need to make sure that any $q$ servers can reconstruct the entire matrix $\mathbf{A}$. By the erasure-correcting properties of the MDS code, this can be satisfied when at least $m$ distinct rows from $\{\cv_i\}_{i=1}^{r_1m}$ are stored at any $q$ servers. We will see next that, when inequality \eqref{cons:numr} holds, this condition is satisfied by using also the redundancy created by the repetition code. To this end, we distinguish two cases in terms of the choice of rates $r_1$ and $r_2$. 

\emph{Case 1:} For large repetition redundancy, i.e., for $r_2>K-q$, all the $r_1m$ encoded rows of $\cv$ are stored at the $q$ servers in $\mathcal{Q}$. This is because, the number $K-q$ of straggling servers is smaller than the number of repetitions $r_2$ and hence the rows in each subset $\mathcal{C}_{\mathcal{K}}$ are stored by at least $r_2-(K-q)$ servers in $\mathcal{Q}$. As a result, inequality \eqref{cons:numr} holds immediately where we have $\binom{K-q}{r_2}=0$ and $r_1\geq 1$. This means that a large rate $r_2 > K-q$ that satisfies conditions \eqref{cons:lim} and \eqref{cons:stor} can guarantee the reconstruction of $\mathbf{A}$. An instance of this case is the example in Section~\ref{example}. 

\emph{Case 2:} For small repetition redundancy, i.e., for $r_2\leq K-q$, the number of subsets $\{\mathcal{C}_{\mathcal{K}}:\mathcal{K}\subseteq [K]\backslash\mathcal{Q}\}$ that are exclusively stored at the $K-q$ straggling servers in set $[K]\backslash\mathcal{Q}$ is given as $\binom{K-q}{r_2}$, while the other $f=\binom{K}{r_2}-\binom{K-q}{r_2}$ subsets $\{\mathcal{C}_{\mathcal{K}}:\mathcal{K} \cap\mathcal{Q}\neq\varnothing\}$ are stored by at least one of the servers in $\mathcal{Q}$. These subsets include $|\mathcal{C}_{\mathcal{K}}| f$ encoded rows. In order to satisfy the reconstruction requirement, we hence have inequality $|\mathcal{C}_{\mathcal{K}}| f\geq m$, which coincide with condition \eqref{cons:numr}.


\emph{Shuffling phase:}
Each server $k$ in $\mathcal{Q}$ is assigned an arbitrary subset $\mathcal{R}_k$ of $|\mathcal{R}_k|=N/q$ disjoint vectors out of the $N$ vectors $\{\yv_i\}_{i=1}^{N}$. For each assigned vector $\yv_i=A\xv_i$, $i\in\mathcal{R}_k$, server $k$ has only available $|\mathcal{C}_k|$ IVs. Hence, $m-|\mathcal{C}_k|=m(1-r_1r_2/K)$ IVs are needed in order to reduce each vector in $\mathcal{R}_k$, yielding a total number $m(1-r_1r_2/K)N/q$ IVs needed for each server in $\mathcal{Q}$.  


As discussed, the subsets of encoded rows that are stored across the $q$ servers in $\mathcal{Q}$ are described by the set $\{\mathcal{C}_{\mathcal{K}}:|\mathcal{K} \cap\mathcal{Q}|\neq 0\}$. Note that this is also true for Case 1 defined above, since all subsets $\mathcal{K}$ satisfy $|\mathcal{K} \cap\mathcal{Q}|\neq 0$ in this case. For each subset $\mathcal{C}_{\mathcal{K}}$ with $|\mathcal{K} \cap\mathcal{Q}|\neq 0$, the product $\cv\mathbf{X}$ with $\cv\in\mathcal{C}_{\mathcal{K}}$ is available at $|\mathcal{K}\cap\mathcal{Q}|$ of the $q$ servers in $\mathcal{Q}$. By construction, the number of servers in set $\mathcal{K}\cap\mathcal{Q}$ ranges in the interval $[s_{min}: \min\{q,r_2\}]$, where we have defined $s_{min}=\max\{r_2-(K-q),1\}$. This number is referred to as the redundancy of the IVs $\cv \xv_i$, with $\cv\in\mathcal{C}_{\mathcal{K}}$ and any $\xv_i\in [N]$.


%

The coded multicasting approach in \cite{LMA:16} is now performed to deliver $m(1-r_1r_2/K)N/q$ IVs to each server in $\mathcal{Q}$ in the shuffling phase. Following \cite{LMA:16}, there are $s_{max}-s_q+1$ different phases, indexed as $i=s_{max},s_{max}-1,\cdots,s_q$, where we have defined $s_{max}=\min\{q-1,r_2\}$ and $s_q=\inf\{s:\sum_{i=s}^{s_{max}}\binom{q-1}{i}\binom{K-q}{r_2-i}/\big(\frac{1}{r_1}\binom{K}{r_2}\big)=\sum_{i=s}^{s_{max}}B_i\leq 1-r_1r_2/K\}$. In each phase $i$, the multicasting gain is given by $i\in[s_q,s_{max}]$ in the sense that multicast messages are sent to $i$ servers at a time.

For each phase $i=s_{max},s_{max}-1,\cdots,s_q$, a set of multicasting messages are sent within groups of $i+1$ servers in $\mathcal{Q}$, so as to deliver the needed IVs that are available at subsets of $i$ servers in each group. There are $\binom{q}{i+1}$ such subsets of $i+1$ servers in $\mathcal{Q}$, and each group operates in the same way by sending multicasting messages from one server to the other $i$ servers.

Specifically, phase $i$ is used to share IVs with redundancy $i$, i.e., which corresponds to rows in the set $\{\mathcal{C}_{\mathcal{K}}:|\mathcal{K} \cap\mathcal{Q}|=i\}$. In each phase $i$, there are $\binom{q}{i+1}$ groups $\mathcal{S}_i \subseteq[K]$ of $i+1$ non-straggling servers. Each server $k$ in a group $\mathcal{S}_i$ needs to receive IVs $\{\cv\xv_t: \cv\in \{\mathcal{C}_{\mathcal{K}}:\mathcal{S}_i\backslash\{k\}\subseteq\mathcal{K}, |\mathcal{K}\cap \mathcal{Q}|=i\}, t\in\mathcal{R}_{k}\}$ that are available at all $i$ servers in $\mathcal{S}_i\backslash\{k\}$. The number of such IVs is given as $\binom{K-q}{r_2-i}|\mathcal{C}_{\mathcal{K}}||\mathcal{R}_{k}|$, where $\binom{K-q}{r_2-i}$ is the number of subset $\mathcal{C}_{\mathcal{K}}$ and $|\mathcal{C}_{\mathcal{K}}|=(r_1m)/\binom{K}{r_2}$ is the number of encoded rows in each $\mathcal{C}_{\mathcal{K}}$. Each server in $\mathcal{S}_i\backslash\{k\}$ delivers a fraction $1/i$ of the IVs. In this way, each server $k$ delivers $\binom{K-q}{r_2-i}|\mathcal{C}_{\mathcal{K}}||\mathcal{R}_{k}|/i$ required IVs for each server in $\mathcal{S}_i\backslash\{k\}$. To this end, server $k$ creates $\binom{K-q}{r_2-i}|\mathcal{C}_{\mathcal{K}}||\mathcal{R}_{k}|/i$ multicasting messages by summing $i$ IVs, one for each server in $\mathcal{S}_i\backslash\{k\}$.

Based on the discussion above, the total number of delivered multicasting messages is given as $\binom{q}{i+1}(i+1)\binom{K-q}{r_2-q}|\mathcal{C}_{\mathcal{K}}||\mathcal{R}_{k}|/i$. Furthermore, each server $k$ in set $\mathcal{Q}$ can recover $\binom{q-1}{i}\binom{K-q}{r_2-q}|\mathcal{C}_{\mathcal{K}}||\mathcal{R}_{k}|$ IVs, since each server is included in $\binom{q-1}{i}$ groups $\mathcal{S}_i$. Combining all the $s_{max}-s_q+1$ phases, the communication load is given as  
\begin{align}
L_1=\sum_{i=s_q}^{s_{max}}\frac{\binom{q}{i+1}(i+1)\binom{K-q}{r_2-q}|\mathcal{C}_{\mathcal{K}}||\mathcal{R}_{k}|}{im}=\sum_{i=s_q}^{s_{max}}N\frac{\binom{q-1}{i}\binom{K-q}{r_2-i}}{i\frac{1}{r_1}\binom{K}{r_2}}=\sum_{i=s_q}^{s_{max}}N\frac{B_i}{i},
\end{align}
and the total number of IVs each server $k$ receives is given as 
\begin{align} \label{eq:n}
z=\sum_{i=s_q}^{s_{max}}\binom{q-1}{i}\binom{K-q}{r_2-q}|\mathcal{C}_{\mathcal{K}}||\mathcal{R}_{k}|=\sum_{i=s_q}^{s_{max}} \frac{m\binom{q-1}{i}\binom{K-q}{r_2-i}}{\binom{K}{r_2}\frac{1}{r_1}}\frac{N}{q}=\sum_{i=s_q}^{s_{max}}  B_i\frac{mN}{q},
\end{align}
where we have the inequality $z\leq m(1-r_1r_2/K)N/q$ due to the definition of $s_q$. Hence, at the end of the $s_{max}-s_q+1$ phases, each server in $\mathcal{Q}$ still needs $l=m(1-r_1r_2/K)N/q-z$ IVs, which are delivered by considering two cases. 
 
\emph{Case 1:} $s_q=s_{min}$. As discussed, for each server $k$, the IVs in set $\big\{\cv\xv_t: \cv\in \{\mathcal{C}_{\mathcal{K}}:|\mathcal{K}\cap\mathcal{Q}\backslash\{k\}|\neq 0\}, t\in\mathcal{R}_k \big\}$ that are available across the other servers in $\mathcal{Q}\backslash\{k\}$ has cardinality $\sum_{i=s_{min}}^{s_{max}} B_imN/q$. This is obtained by summing up the IVs in the set with each redundancy $|\mathcal{K}\cap\mathcal{Q}\backslash\{k\}|$, ranging in the interval $[s_{min}:s_{max}]$. Since each server needs $m(1-r_1r_2/K)N/q$ IVs, the inequality $\sum_{i=s_{min}}^{s_{max}} B_imN/q \geq m(1-r_1r_2/K)N/q$ holds immediately. Combining with \eqref{eq:n} and the inequality $z\leq m(1-r_1r_2/K)N/q$, when $s_q=s_{min}$, we have the equality $z=m(1-r_1r_2/K)N/q$, i.e., $l=0$. This implies that the computed IVs for each server across the $q$ servers is $mN/q$, i.e., the number of distinct encoded rows stored across the $q$ servers is exactly $m$. As a result, the communication is finished with the load $L(q)=L_1$ for this case.

\emph{Case 2:} $s_q>s_{min}$. The $l$ IVs are sent by considering an additional phase indexed as $i=s_q-1$. Due to the definition of $s_q$, we have the inequality $\sum_{j=s_q-1}^{s_{max}}  B_j\geq (1-r_1r_2/K)$. This implies that 
the number of IVs in the set $\{\cv\xv_t: \cv\in \{\mathcal{C}_{\mathcal{K}}:\mathcal{S}_i\backslash\{k\}\subseteq\mathcal{K}, |\mathcal{K}\cap \mathcal{Q}|=i\}, t\in\mathcal{R}_{k}\}$ is larger than $l$, and hence only $l$ from them need to be received. Splitting this number equally into $i$ fractions, coded multicasting with the gain $s_q-1$ is performed in the same way as discussed above. Thus, the overall communication load is given as 
\begin{align}
L(q)=L_1+\frac{lq}{m(s_q-1)}=L_1+ \frac{mN(1-r_1r_2/K)-nq}{m(s_q-1)}.
\end{align}

\section{concluding Remarks} \label{sec:4}

  
In this paper, we have studied the latency-communication trade-off in a map-shuffle-reduce network in the presence of straggling servers. The key observation is that the redundancy of the repetition code in the scheme proposed in \cite{LMA:16} can be used not only to accelerate communications in the shuffling phase, but also to correct erasures caused by the straggling servers. This approach was shown to improve the latency-communication trade-off derived in \cite{LMA:16}.

\begin{appendix} \label{prove}
We now prove that $r_1=K/q$ and $r_2=\lfloor q\mu \rfloor$ is a feasible solution that satisfies conditions $\eqref{condition}$. The first two conditions are immediately verified. For condition $\eqref{cons:numr}$, we have
\begin{align}
\frac{\binom{K-q}{r_2}}{\binom{K}{r_2}}=\frac{\prod_{j=1}^{q}(K-q-r_2+j)}{\prod_{j=1}^{q}(K-q+j) }\stackrel{(a)}{=}\frac{\prod_{j=1}^{r_2}(K-q-r_2+j)}{\prod_{j=1}^{r_2}(K-r_2+j) }\stackrel{(b)}{\leq} \frac{K-q}{K},
\end{align}
where equality $(a)$ holds by removing only the common elements in numerator and denominator, and inequality $(b)$ holds because for any $j\in[r_2]$, we have the inequality $K-q-r_2+j \leq K-r_2+j$ and we have $(K-q)\leq K$ for $j=r_2$. Hence, the inequality $\binom{K}{r_2}-\binom{K-q}{r_2}\geq (q/K)\binom{K}{r_2}=(1/r_1)\binom{K}{r_2}$ holds, implying that condition $\eqref{cons:numr}$ is satisfied. 
\end{appendix}

\section*{Acknowledgements}
Jingjing Zhang and Osvaldo Simeone have received funding from the European Research Council (ERC) under the European Union's Horizon 2020 Research and Innovation Programme (Grant Agreement No. 725731).

\bibliographystyle{IEEEtran}
\bibliography{IEEEabrv,final_refs}

\end{document}